\documentclass[10pt,conference]{IEEEtran}

\usepackage[dvips]{color}
\usepackage{epsf}

\usepackage{bbm}
\usepackage{anyfontsize}
\usepackage{t1enc}
\usepackage{amsmath}
\usepackage{amsxtra}
\usepackage{graphicx}
\usepackage{epsfig}
\usepackage{latexsym}
\usepackage{amsfonts}
\usepackage{rawfonts}
\usepackage[latin1]{inputenc}
\usepackage[T1]{fontenc}
\usepackage{calc}
\usepackage{url}
\usepackage{enumerate}
\usepackage{color}
\usepackage{amssymb}
\usepackage{upref}
\usepackage{epic,eepic}
\usepackage{times}
\usepackage{cite}
\usepackage{dsfont}
\usepackage{mathrsfs}
\usepackage{verbatim}
\usepackage{float}
\usepackage{chngcntr}
\usepackage{bbm}
\usepackage{algorithm2e}
\usepackage{hyperref}
\hypersetup{colorlinks, citecolor={[rgb]{0.83,0.12,0.04}}, linkcolor={[rgb]{0.83,0.12,0.04}}, urlcolor={[rgb]{0.00,0.22,0.44}} }

\DeclareMathOperator*{\minim}{minimize}
\DeclareMathOperator*{\maxim}{maximize}

\newtheorem{theorem}{\bf{ Theorem}}

\newtheorem{definition}{\bf{ Definition}}
\newtheorem{remark}{\bf{Remark}}
\newcommand{\qed}{\nobreak \ifvmode \relax \else
  \ifdim\lastskip<1.5em \hskip-\lastskip
  \hskip1.5em plus0em minus0.5em \fi \nobreak
  \vrule height0.75em width0.5em depth0.25em\fi}

\newlength{\totlinewidth}

\newcounter{substep}

  {\end{list}}

\newlength{\aligntop}
\newlength{\alignbot}
 \makeatletter
\renewenvironment{align}{%
  \vspace{\aligntop}
  \start@align\@ne\st@rredfalse\m@ne
}{%
  \math@cr \black@\totwidth@
  \egroup
  \ifingather@
    \restorealignstate@
    \egroup
    \nonumber
    \ifnum0=`{\fi\iffalse}\fi
  \else
    $$%
  \fi
  \ignorespacesafterend%
  \vspace{\alignbot}\par\noindent
} \makeatother

\IEEEoverridecommandlockouts

\begin{document}
\title{\huge Context-Aware Scheduling of Joint Millimeter Wave and Microwave Resources for Dual-Mode Base Stations}\vspace{0em}
\author{
\authorblockN{Omid Semiari$^{\dag}$, Walid Saad$^{\dag}$, and Mehdi Bennis$^\ddag$}\\\vspace*{0em}
\authorblockA{\small $^{\dag}$Wireless@VT, Bradley Department of Electrical and Computer Engineering, Virginia Tech, Blacksburg, VA, USA, \\Emails: \protect\url{{osemiari,walids}@vt.edu}\\
\small $^\ddag$ Centre for Wireless Communications, University of Oulu, Finland, Email: \url{bennis@ee.oulu.fi}
}\vspace*{-3.1em}
    \thanks{This research was supported by the U.S. National Science Foundation under Grants CNS-1460316 and CNS-1526844.}%
  }
%
\maketitle
\begin{abstract}
One of the most promising approaches to overcome the drastic channel variations of millimeter wave (mmW) communications is to deploy dual-mode base stations that integrate both mmW and microwave ($\mu$W) frequencies. Reaping the benefits of a dual-mode operation requires scheduling mechanisms that can allocate resources efficiently and jointly at both frequency bands. In this paper, a novel resource allocation framework is proposed that exploits users' context, in terms of user application (UA) delay requirements, to maximize the quality-of-service (QoS) of a dual-mode base station. In particular, such a context-aware approach enables the network to dynamically schedule UAs, instead of users, thus providing more precise delay guarantees and a more efficient exploitation of the mmW resources. The scheduling of UAs is formulated as a one-to-many matching problem between UAs and resources and a novel algorithm is proposed to solve it. The proposed algorithm is shown to converge to a two-sided stable matching between UAs and network resources. Simulation results show that the proposed approach outperforms classical CSI-based scheduling in terms of the per UA QoS, yielding up to $36 \%$ improvement. The results also show that exploiting mmW resources provides significant traffic offloads reaching up to $43 \%$ from $\mu$W band. \vspace{-.2cm}
\end{abstract}
\section{Introduction} \vspace{-0cm}

Communication at high frequency, millimeter wave (mmW) bands is an effective way to boost the performance of 5G cellular networks \cite{6736746,Rangan14}. However, field measurements \cite{Rangan14} have shown that the availability of mmW links can be highly intermittent, due to blockage by various obstacles. Therefore, meeting quality-of-service (QoS) constraints of delay-sensitive applications, such as HDTV and video conferencing, is challenging at mmW frequencies \cite{shokri,7010537,7010536,5783993, 4689210,7010539,Rangan14}.   


To provide a robust and reliable communication, mmW networks must coexist with small cell LTE networks that operate at the conventional microwave ($\mu$W) band \cite{7010537,7010536,5783993, 4689210,7010539}. In \cite{7010537}, the authors analyze transceiver architectures for dual-mode mmW-$\mu$W networks. The work in \cite{5783993} proposes a mmW-$\mu$W dual-mode architecture used to transmit control and data signals, respectively, at $\mu$W and mmW frequency bands. 

The problem of QoS provisioning for mmW is studied in \cite{7010536,4689210}, and \cite{7010539}. In \cite{7010536}, the authors propose a scheduling scheme that integrates device-to-device mmW links with 4G system to bypass the blocked mmW links. The work in \cite{4689210} presented a mmW system at $60$ GHz for supporting uncompressed high-definition (HD) videos for WLANs. In \cite{7010539}, the authors defined and evaluated important metrics to characterize multimedia QoS, and designed a QoS-aware multimedia scheduling scheme to achieve the trade-off between performance and complexity.

 Although interesting, the first body of work in \cite{7010536,7010537,5783993} does not address the QoS provisioning in mmW-$\mu$W networks. Moreover, \cite{7010537,4689210}, and \cite{7010539} do not consider multi-user scheduling and multiple access in dual-mode networks. In addition, conventional scheduling mechanisms, such as \cite{7010536,4689210}, and \cite{7010539}, identify each user equipment (UE) by a single traffic stream with a certain QoS requirement. In practice, however, recent trends show that users run multiple applications simultaneously, each with a different QoS requirement. Even though the applications at a single device experience the same wireless channel, they may tolerate different delays and QoS, thus, resulting in different user's quality of experience. Accounting for precise, application-specific QoS metrics is particularly important for scheduling mmW resources whose channel is highly variable. In fact, conventional scheduling approaches fail to guarantee the QoS for multiple applications at a single UE.

 With regard to QoS provisioning in dual-mode mmW-$\mu$W networks, it is worthy to note that 1) traffic management mandates a joint scheduling that allocates resources in both frequency bands, and 2) the QoS constraint per user application (UA) can dictate whether the traffic should be served via mmW resources, $\mu$W resources, or both. Such knowledge of the \textit{user's application context information}, is required for a robust and efficient scheduling in networks that incorporate dual-mode small cell base stations (SBSs).



The main contribution of this work is to propose a novel, context-aware resource allocation framework that intelligently allocates mmW and $\mu$W resources of a dual-mode SBS, depending on the specific delay constraints of UAs. This proposed context-aware scheduler allows each user to seamlessly run multiple applications simultaneously, each with certain QoS constraint. We formulate the problem as a two-sided matching game that aims to allocate time-frequency resources to UAs. To solve the game, we propose a novel distributed algorithm that allows UAs to submit requests for network resources based solely on their local information, i.e., tolerable delay and currently perceived network state. We show that the proposed algorithm yields a two-sided stable resource allocation to UAs. Simulation results show that the dual-band scheduler provides significant performance advantages, in terms of traffic offload, efficient mmW exploitation, and improved overall UA delay.


The rest of this paper is organized as follows. Section II presents the problem formulation. Section III presents the proposed matching solution. Simulation results are analyzed in Section IV. Section V concludes the paper.

\section{System Model}
Consider the downlink of a dual-mode small base station (SBS) that operates at microwave ($\mu$W) and millimeter wave (mmW) frequency bands. The coverage area of the SBS is a planar area with radius $r$ centered at $(0,0) \in \mathbb{R}^2$. Moreover, there are $M$ UEs in the set $\mathcal{M}$, distributed randomly and uniformly within the SBS coverage. UEs are equipped with both mmW and $\mu$W RF interfaces which allow them to manage their traffic at both frequency bands. In addition, a UE at distance $d$ from the SBS will experience a LoS mmW connection with probability $\rho$ if $d \leq r$, otherwise, $\rho=0$. 

The antenna arrays of mmW transceivers allow to achieve an overall gain of $\psi(x_1,x_2)$ for a UE located at $(x_1,x_2) \in \mathbb{R}^2$ \cite{Ghosh14}. On the other hand, the transceivers at $\mu$W frequency have conventional single element omni-directional antennas. Furthermore, each UE $m \in \mathcal{M}$ runs $\kappa_m$ UAs. We let $\mathcal{A}$ be the set of all UAs with $A=\sum_{m \in \mathcal{M}}\kappa_m$ as the total number of UAs across all UEs. 

 \subsection{Multiple Access at mmW and $\mu$W Frequency Bands}
 
At mmW band, time division multiple access (TDMA) is used to schedule UAs at time slots of duration $\tau_1$ \cite{7010536}. At each mmW time slot, the SBS will transmit the bits associated with one UA over $K_1$ resource blocks (RBs). To model large-scale channel effects at mmW links, we use the popular model of \cite{Ghosh14}:
 \begin{align}\label{mwpathloss}
 L_1(x_1,x_2)=\beta_1+\alpha_1 10\log_{10}(\sqrt{x_1^2+x_2^2})+\chi,
 \end{align}
 where $L_1(x_1,x_2)$ is the path loss at mmW frequencies for all UAs associated with a UE located at $(x_1,x_2) \in \mathbb{R}^2$. In fact, (\ref{mwpathloss}) is known to be the best linear fit to the propagation measurement in mmW frequency band\cite{Ghosh14}, where $\alpha_1$ is the slope of the fit and $\beta_1$, the intercept parameter, is the pathloss (in dB scale) for $1$ meter of distance. In addition, $\chi$ models the deviation in fitting (in dB scale) which is a Gaussian random variable with zero mean and variance $\xi_1^2$. Overall, the total achievable rate for UA $a$ at time slot $j$ is given by  
\begin{align}\label{rate1}
 R_a(j)\!=\!\!\sum_{k=1}^{K_1}\!w_1\!\!\log_2\!\left(\!\!1\!\!+\!\frac{p_{k,1}\psi(x_1,x_2)|h_{kj}|^2 10^{-\frac{L_1(x_1,x_2)}{10}}}{w_1 N_0}\!\right),
 \end{align}
where $w_1$ is the bandwidth of each RB, $h_{kj}$ is the Rayleigh fading channel coefficient at RB $k$ of slot $j$, and $N_0$ is the noise power. The total transmit power at mmW band, $P_1$, is assumed to be distributed uniformly among all RBs, i.e., $p_{k,1}=P_1/K_1$.
\begin{figure}
\centering
\centerline{\includegraphics[width=7cm]{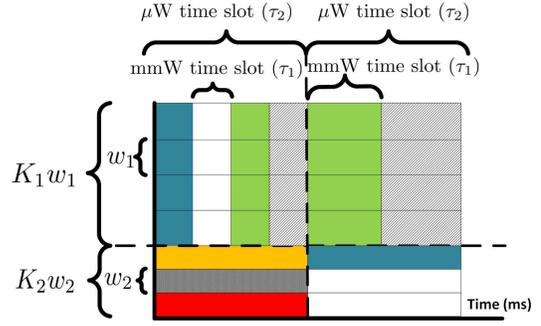}}\vspace{-0.2cm}
\caption{Example of resource allocation of the dual-band configuration. Colors correspond to different UAs that may run at different UEs.}\vspace{-.5cm}
\label{model}
\end{figure}
For the $\mu$W band, we consider orthogonal frequency division multiple access (OFDMA) scheme in which multiple UAs can be scheduled over $K_2$ RBs in the set $\mathcal{K}_2$ at each $\mu$W time slot with duration $\tau_2$. Therefore, the achievable rate for an arbitrary UA $a$ at RB $k$ and time slot $t$ is:
\begin{align}\label{rate2}
R_a(k,t)=w_2\log_2\left(1+\frac{p_{k,2}|g_{kt}|^2 10^{-\frac{L_2(x_1,x_2)}{10}}}{w_2N_0}\right).
\end{align}
where $w_2$ denotes the bandwidth of each RB at $\mu$W band, and $g_{kt}$ is the Rayleigh fading channel at RB $k$ at time-slot $t$. Moreover, the total transmit power at $\mu$W band, $P_2$, is assumed to be distributed uniformly among all RBs, i.e., $p_{k,2}=P_2/K_2$. The path loss $L_2(x_1,x_2)$ follows the log-distance model, similar to (\ref{mwpathloss}), with parameters $\alpha_2$, $\beta_2$, and $\xi_2$ adapted for $\mu$W band.  

Hereinafter, unless otherwise specified, we use ``\emph{time slot'' to refer to the $\mu$W time slot}. The scheduler allocates resources to UAs at the beginning of each time slot which remains unchanged for the next $\tau_2$ seconds. Since mmW operates at very high frequencies, its channel coherence time will be relatively smaller than that of the $\mu$W frequencies \cite{7010539}. Therefore, we let $\tau_1=\tau_2/J(t)$, where $J(t)$ is the number of UAs that will be scheduled in the mmW band at time slot $t$. 


The proposed dual-band multiple access scheme is shown in Fig. \ref{model}, where each color identifies a single UA.  
\subsection{Traffic Model}
We assume a non-full buffer traffic model, where each UA has $B$ bits of data to transmit. Once $B$ bits are transmitted, the corresponding UA will be removed from the scheduling. 
Each UA has an application-specific tolerable delay which specifies its QoS class. 
\begin{definition}
The \textit{QoS class}, $\mathcal{A}_t$, is defined as the set of all UAs over all UEs that can tolerate a packet transmission delay of $t$ time slots.
\end{definition} 

Our system has a total of $T$ QoS classes, $\bigcup_{t=1}^T \mathcal{A}_t=\mathcal{A}$, and $\mathcal{A}_t \cap \mathcal{A}_{t'}=\emptyset, t \neq t'$. Due to system resource constraints, not all UAs can be served instantly, thus,  access delay may occur to UAs.  In fact, to transmit a data stream of size $B$ bits to UA $a \in \mathcal{A}_t$, an average data rate of $B/t\tau_2$ during $t$ consecutive time slots is needed. The context information for all UAs is captured by the set $\mathcal{C}=\{\mathcal{A}_1,\mathcal{A}_2,...,\mathcal{A}_T\}$.


\subsection{Problem Formulation}
At each time slot $t$, a scheduling decision $\pi_{t}$ assigns time-frequency resources to UAs, at both mmW and $\mu$W bands. The scheduler takes the context information $\mathcal{C}$ and achievable rates $R_a(k,t)$ and $R_a(j)$, for all $a \in \mathcal{A}$, $k=1,...,K_2$, and $j=1,...,J(t)$, as inputs and outputs integer variables $x_{akt}$ and $y_{ajt}$. $x_{akt}=1$, if RB $k$ at time slot $t$ is allocated to UA $a$ and $x_{akt}=0$, otherwise. In addition, $y_{ajt}=1$, if mmW time slot $j$ at slot $t$ is allocated to UA $a$, otherwise $y_{ajt}=0$. 

The scheduling decision at a given slot $t$ depends on scheduling decisions at previous time slots $\{\pi_1, \pi_2, ..., \pi_{t-1}\}$. Thus, we define $\pi=\{\pi_1, \pi_2, ..., \pi_t, ...,\pi_M \} \in \boldsymbol{\Pi}$ as a \textit{scheduling policy}, where $\boldsymbol{\Pi}$ is the set of all possible scheduling policies. For a given policy $\pi$, the average rate for UA $a$ after time slot $t$, $\bar{R}_{\pi}(a,t)$, is
\begin{align}\label{eqRbar}
\bar{R}_{\pi}(a,t)\!\!=\!\!\frac{1}{t\tau_2}\!\sum_{t'=1}^{t}\!\!\left[\tau_2\sum\limits_{k=1}^{K_2} R_a(k,t')x_{akt'}\!+\!\tau_1\sum\limits_{j=1}^{J(t')} R_a(j)y_{ajt'}\right]\!.
\end{align}
Next, we use (\ref{eqRbar}) to formally define $\mathbbm{1}_a(\pi)$ as the QoS criterion for $a \in \mathcal{A}_t$ as follow
\begin{align}\label{lambda2}
\mathbbm{1}_a(\pi) = \begin{cases}
0             & \text{if}\,\, \bar{R}_{\pi}(a,t)\ge\frac{B}{t\tau_2},\\
1             & \text{otherwise},
\end{cases}
\end{align}
where $\mathbbm{1}_a(\pi)=0$ indicates that enough resources are allocated to UA $a$, to receive $B$ bits within $t$ slots.

Directional transmissions at mmW band compel the SBS to steer the beam toward one UE at any given mmW time slot \cite{6736746}. Hence, without prior information of the channel for all UAs, the scheduler cannot schedule UAs based on the achievable rates, $R_a(j)$. Therefore, the scheduler follows an opportunistic round-robin (RR) scheme\footnote{\vspace{-1cm} Other random access schemes can be readily accommodated in our model.} at mmW band to allocate time slots to UAs. Therefore, the dual-mode SBS must be able to exploit the $\mu$W resources in order to meet the QoS constraints of UAs, i.e., to minimize $\sum_{t=1}^{T}\sum_{a \in \mathcal{A}_t} \mathbbm{1}_a(\pi)$. 

To this end, we focus on a subset of scheduling policies $\pi \in \boldsymbol{\Pi}^c \subseteq \boldsymbol{\Pi}$ that schedules UAs in $\mathcal{A}_t$ at mmW band for the first $t-1$ time slots. Furthermore, the scheduler satisfies their required average rate, $\bar{R}_{\pi}(a,t)$, by efficiently allocating the resources of $\mu$W band at time slot $t$ to UAs in $\mathcal{A}_t$. Hence, given $\pi \in \boldsymbol{\Pi}^c$, we can write (\ref{eqRbar}) as 
\begin{align}\label{eqRbar2}
\bar{R}_{\pi}(a,t)\!=\!\frac{1}{t\tau_2}\!\left[\tau_2\sum\limits_{k=1}^{K_2} R_a(k,t)x_{akt}+B_{\pi}(a,t-1)\right]\!,
\end{align}
where $B_{\pi}(a,t-1)$ is the total number of transmitted bits for UA $a$ up until time slot $t$, i.e., 
\begin{align}\label{eqRbar3}
B_{\pi}(a,t-1)= \tau_1\sum_{t'=1}^{t-1}\sum\limits_{j=1}^{J(t')} R_a(j)y_{ajt'}.
\end{align}
Now, we formulate the problem as follows, for all $t=1,...,T$:
\begin{IEEEeqnarray}{rCl}\label{eq4}
&&\minim_{\pi \in \boldsymbol{\Pi}^c} \,\, \sum_{t=1}^{T}\sum_{a \in \mathcal{A}_t} \mathbbm{1}_a(\pi), \IEEEyessubnumber\label{1a}\\
\textrm{s.t.} 
&& \sum_{a=1}^{A}B_{\pi}(a,T) \leq B_{\text{tot}},\IEEEyessubnumber\label{1c}\\
&& \sum_{j=1}^{J(t)}y_{ajt} \leq 1,\forall a \in \mathcal{A},\IEEEyessubnumber\label{1d}\\
&& \sum_{a=1}^{A}y_{ajt}\leq 1,  j=1,...,J(t),\IEEEyessubnumber\label{1e}\\
&& \sum_{k=1}^{K_2}x_{akt} \leq K_2, \forall a \in \mathcal{A}, \IEEEyessubnumber\label{1f}\\
&& \sum_{a=1}^{A}x_{akt}\leq 1,k=1,...,K_2, \IEEEyessubnumber\label{1g}\\
&& \sum_{k=1}^{K_2}\sum_{j=1}^{J(t)}x_{akt}y_{ajt}=0,\forall a \in \mathcal{A}.\,\IEEEyessubnumber\label{1h}
\end{IEEEeqnarray}
In (\ref{1a}), the objective is to maximize the QoS for UAs, using both mmW and $\mu$W resources. From (\ref{lambda2}), we can see that the objective function incorporates the context information $\mathcal{C}$. (\ref{1c}) implies that the total transmitted bits after time slot $T$, $B_{\pi}(M)$, must be less or equal to the total load $B_{\text{tot}}=B \cdot A$. (\ref{1d})-(\ref{1e}) ensure orthogonal time-slot allocation for the mmW band. (\ref{1f})-(\ref{1g}) guarantee orthogonal RB allocation at $\mu$W with OFDMA. Furthermore, (\ref{1h}) implies that a single UA cannot be simultaneously assigned to both mmW and $\mu$W bands.  
Next, we propose a framework to solve the optimization problem. 

\section{Context-Aware Scheduling as a Matching Game}
In (\ref{eqRbar2}) and (\ref{1a}), we observe that $B_{\pi}(a,t-1)$ is the only required information from previous time slots to schedule UA $a \in \mathcal{A}_t$ at time slot $t$. Thus, at each time slot $t$ we have:
\begin{remark}
The optimal scheduling decision $\pi_t^*$ at time slot $t$ is the solution of
\begin{IEEEeqnarray}{rCl}\label{lem2}
&&\maxim_{\pi_t \in \boldsymbol{\Pi}^c} \,\, \sum_{a \in \mathcal{A}_t}\sum_{k=1}^{K_2}R_a(k,t)x_{akt}  \IEEEyessubnumber\label{lem2a}\\
\textrm{s.t.} 
&& \sum_{k=1}^{K_2}\!R_a(k,t)x_{akt}\!  \leq \!B\!-\!B_{\pi}(a,\!t\!-\!1),\,\,\,\,\,\,\forall a \in \mathcal{A}_t,\IEEEyessubnumber\label{lem2b}\\
&& (\ref{1c})-(\ref{1h}).\IEEEyessubnumber\label{lem2e}
\end{IEEEeqnarray} 
\end{remark}
The downlink scheduling problem in (\ref{lem2a})-(\ref{lem2e}) is a combinatorial problem of matching users to resources which does not admit a closed-form solution and has an exponential complexity\cite{4036195}. 


\subsection{Scheduling as a Matching Game: Preliminaries}
To solve the resource allocation problem in (\ref{lem2a})-(\ref{lem2e}), we propose a novel solution based on matching theory, a mathematical framework that provides a decentralized solution with tractable complexity for combinatorial problems, such as the one in (\ref{lem2a})-(\ref{lem2e}) \cite{Roth92,eduard11}. A \emph{matching game} is defined as a two-sided assignment problem between two disjoint sets of players in which the players of each set are interested to be matched to the players of the other set, according to \textit{preference relations}. At each time slot $t$ of our scheduling problem, $\mathcal{K}_2$ and $\mathcal{A}_t$ are the two sets of players. A preference relation $\succ$ is defined as a complete, reflexive, and transitive binary relation between the elements of a given set. Here, we let $\succ_a$ be the preference relation of UA $a$ and denote $k\succ_a k'$, if player $a$ prefers RB $k$ more than $k'$. Similarly, we use $\succ_k$ to denote the preference relation of RB $k \in \mathcal{K}_2$.

In the proposed scheduling problem, the preference relations of UAs depend on both the rate and the QoS constraint. Matching theory allows to specify preference relations, by defining individual utility functions for UAs and SBS resources. In our scheduling game, the SBS will naturally control the preferences of all resources.

\subsection{Dual-mode Scheduling as a Matching Game}
Each scheduling decision $\pi_t$ determines the allocation of RBs to UAs at time slot $t$. Thus, the scheduling problem can be defined as a \textit{one-to-many matching game}:
\begin{definition}
Given two disjoint finite sets of players $\mathcal{A}_t$ and $\mathcal{K}_2$, the scheduling decision at time slot $t$, $\pi_t$, can be defined as a \textit{matching relation}, $\mu_t:\mathcal{A}_t  \rightarrow \mathcal{K}_2$ that satisfies 1) $\forall a \in \mathcal{A}_t, \mu_t(a) \subseteq \mathcal{K}_2$, 2) $\forall k \in \mathcal{K}_2, \mu_t(k) \in \mathcal{A}_t$, and 3) $\mu_t(k)=a$, if and only if  $k \in \mu_t(a)$.
\end{definition}    
In fact, $\mu_t(k)=a$ implies that $x_{akt}=1$, otherwise $x_{akt}=0$. Therefore, $\mu_t$ can be viewed as a scheduling decision $\pi_t$ that determines the allocation at $\mu$W band. One can easily see from the above definition that the proposed matching game inherently satisfies the constraints in (\ref{1f})-(\ref{1g}). Next, we need to define suitable utility functions to determine the preference profiles of UAs and RBs. Given matching $\mu_t$, we define the utility of UA $a$ for $k \in \mathcal{K}_2$ at time slot $t$ as:
{\fontsize{8.95}{20}
\begin{align}\label{util1}
\Psi_a(k,t';\mu_t)\!\!=\!
&\begin{cases}
0             &\text{if}\,\, a \notin \mathcal{A}_t \,\,\,\,\text{or}\\
&\sum\limits_{k' \in \mu_t(a)}\!\!\!R_a(k',t)\tau_2\!>\!B-B_{\pi}(a,t'-1),\\
R_a(k,t) &\text{otherwise.} 
\end{cases}
\end{align}}
The utility of $\mu$W RBs $k \in \mathcal{K}_2$ for UA $a \in \mathcal{A}_t$ is simply the rate:
\begin{align}\label{util2}
\Phi_k(a,t)=R_a(k,t).
\end{align}
Using these utilities, the preference relations of UAs and RBs at a given time slot $t$ are:
\begin{IEEEeqnarray}{rCl}\label{prefer}
k \succ_a k' &&\Leftrightarrow \Psi_a(k,t;\mu_t) \geq \Psi_a(k',t;\mu_t) \label{prefer1}\\
a \succ_k a' &&\Leftrightarrow \Phi_k(a,t) \geq \Phi_k(a',t) \label{prefer2}, 
\end{IEEEeqnarray}
for $\forall a, a' \in \mathcal{A}$, and $\forall k,k' \in \mathcal{K}_2$. We note that (\ref{util1}) and (\ref{prefer1}) depend on the context information $\mathcal{C}=\{\mathcal{A}_1,\mathcal{A}_2,...,\mathcal{A}_T\}$, while (\ref{util2}) and (\ref{prefer2}) rely only on the channel state information. Thus, the SBS will not need to know the delay tolerance of each UA, making the matching game suitable for distributed implementations.  
\subsection{Proposed Context-aware Scheduling Algorithm}
{\fontsize{8}{20}
\begin{table}[!t]
\centering
\caption{
	\vspace*{-0em}Proposed Context-Aware Scheduling Algorithm}\vspace*{-0.9em}
\begin{tabular}{p{7.5 cm}}
\hline \vspace*{-0em}
\textbf{Inputs:}\,\,$\mathcal{C}$, $\mathcal{A}$, $\mathcal{K}_2$, $B$.\\
\hspace*{1em}\textit{Initialize:}   \vspace*{0em}
$t=1$; $B_{\pi}(a,0)=0$; $\mathcal{K}_a=\mathcal{K}_2,\,\,\, \forall a \in \mathcal{A}$.\\
\hspace*{-0.5em} \While{$t \leq T$}{\eIf{$a \in \mathcal{A}_t$}{
		\begin{itemize}
		\item[1.] Find $R_a(k,t)$ for $\forall k \in \mathcal{K}_2$ and $B_{\pi}(a,t-1)$.
		\item[2.] Update the preference ordering of UAs and RBs, using (\ref{prefer1}), (\ref{prefer2}).
		\item[3.] Using $\succ_a$, UA $a$ applies for the most preferred RB in $\mathcal{K}_a$.
		\item[4.] Each RB $k$ accepts the most preferred UA, based on $\succ_k$, among new applicants plus $\mu_t(k)$, and rejects the rest. Next, $k$ is removed from applicants' $\mathcal{K}_a$ sets.
		\item[5.] Each UA $a$ calculates $\bar{R}_{\pi}(a,t)$ and updates $\succ_a$.
		\end{itemize}
		\textbf{repeat} steps 2 to 5 \textbf{until} $\mathbbm{1}_{a}(\pi)\!=\!0$ or $\mathcal{K}_a\!=\!\emptyset$, $\forall a \in \!\mathcal{A}_t$. 
		\begin{itemize}
		\item[6.] Data transmission occurs for each UA $a$ over $\mu_t(a)$ RBs.
		\end{itemize}}{\begin{itemize}
		\item[I.] UAs $a \!\!\in\!\! \bigcup_{t'>t}\mathcal{A}_{t'}$ send a request to SBS for mmW resource.
		\item[II.] SBS sets $J(t)$ to the number of received requests (from UAs with corresponding LoS UEs) and adjusts $\tau_1=\tau_2/J(t)$.
		\item[III.] SBS allocates mmW time-slots $j=1,...,J(t)$ to applicants based on RR and updates $y_{ajt}$ variables.
		\item[IV.] If $y_{ajt}=1$, SBS transmits data at time-slot $j$ to UA $a$. UA updates $B_{\pi}(a,t)$.
		\end{itemize}}
	$t\gets t+1$}\vspace*{-0cm}
\hspace*{0em}\textbf{Output:}\,\,Stable scheduling policy $\pi^*$\vspace*{0em}\\
\hline
\end{tabular}\label{tab:algo}\vspace{-0.5cm}
\end{table}}

To solve the proposed game and find a suitable outcome, we use the concept of two-sided \textit{stable matching} between UAs and RBs, defined as follows \cite{Roth92}:
\begin{definition}
A pair $(a,k) \notin \mu_t$ is said to be a \textit{blocking pair} of the matching $\mu_t$, if and only if $a \succ_{k} \mu_t(k)$ and $k \succ_a \mu_t(a)$.
Matching $\mu_t$ is \textit{stable}, if there is no blocking pair.
\end{definition}
Under a stable matching $\mu_t$, one can ensure that the scheduler will not reallocate the RBs to other UAs. In fact, stability is a key requirement for distributed scheduling to ensure that UAs will not deviate from the solution that guarantees the QoS.

For conventional matching problems, the popular \emph{deferred acceptance (DA)} algorithm is used to find a stable matching \cite{Roth92,eduard11}. However, DA cannot be applied directly to our problem because it assumes that the quota for each UA is fixed. The quota is the maximum number of RBs that a UA can be matched to. In our problem, however, quotas cannot be predetermined, since the number of required RBs to satisfy the QoS constraint of a UA, in (\ref{lem2b}), depends on the channel quality at each RB.


In fact, the adopted utility functions in (\ref{util1}) depend on the current state of the matching. Due to the dependency of UAs' preferences to the state of the matching, i.e. $x_{akt}$ variables, the proposed game can be classified as a \textit{matching game with externalities}. For matching games with externalities, DA may not converge to a two-sided stable matching. Therefore, a new algorithm must be found to solve the problem.

To this end, we propose the context-aware scheduling algorithm shown in Table \ref{tab:algo}. At each slot $t$, UAs apply for mmW and $\mu$W resources based on their local information. Steps 1-6 find the stable matching $\mu_t$ for UAs that must be scheduled at $\mu$W band. Steps I-IV use RR to allocate mmW slots to UAs. For this algorithm, we have:
\begin{theorem}
The proposed algorithm in Table \ref{tab:algo} yields a two-sided stable matching between UAs and $\mu$W RBs.
\end{theorem}
\begin{proof}
Since at any arbitrary time slot $t$, UAs $a \in \mathcal{A}_t$ are involved in the matching game and $\mathcal{A}_t \bigcap \mathcal{A}_t'=\emptyset, t \neq t'$, it suffices to prove the two-sided stability of the matching algorithm at time slot $t$. The convergence of the algorithm in Table \ref{tab:algo} at each slot is guaranteed, since a UA never applies for a certain RB twice. Hence, in the worst case, all UAs will apply for all RBs once, which yields $\mathcal{K}_{a}=\emptyset, \forall a \in \mathcal{A}$. Next, we show that once the algorithm converges, the resulting matching between UAs and RBs is two-sided stable. Assume that there exists a pair $(a,k) \notin \mu_t$ that blocks $\mu_t$. Since the algorithm has converged, we can conclude that at least one of the following cases is true about $a$: $\mathbbm{1}_a=0$ , or $\mathcal{K}_{a}=\emptyset$.

The first case, $\mathbbm{1}_a=0$ implies that $a$ does not need to add more RBs to $\mu_t(a)$. In addition, $a$ would not replace any of $k' \in \mu(a)$ with $k$, since $k' \succ_{a} k$. Otherwise, $a$ would apply earlier for $k$. If $a$ has applied for $k$ and got rejected, this means $\mu(k) \succ_k a$, which contradicts $(a,k)$ to be a blocking pair. Analogous to the first case, $\mathcal{K}_{a}=\emptyset$ implies that $a$ has got rejected by $k$, which means $\mu(k) \succ_k a$ and $(a,k)$ cannot be a blocking pair. This proves the theorem.
\end{proof}
\section{Simulation Results}\label{sec:sim}
For simulations, we consider an area with diameter $r=500$ meters with the SBS located at the center. UEs are distributed uniformly within this area with a minimum distance of $5$ meters from the SBS. Each UE has $\kappa$ UAs chosen randomly and uniformly from $T$ QoS classes. The main parameters are summarized in Table \ref{tabsim}. All statistical results are averaged over a large number of independent runs.{\fontsize{7}{20}
\begin{table}[!t]
\centering
\caption{
	\vspace*{-0em}Simulation parameters}\vspace*{-0em}
\begin{tabular}{|c|c|c|}
\hline
\bf{Notation} & \bf{Parameter} & \bf{Value} \\
\hline
$P_1,P_2$ & Transmit power & $1$ W\\
\hline
$(\Omega_1,\Omega_2)$ & Bandwidth & ($1$ GHz, $10$ MHz)\\
\hline
$\omega_1,\omega_2$ & Bandwidth per RB& $480$ KHz\\
\hline
($\xi_1, \xi_2$) & Standard deviation of mmW path loss& ($5.2, 10$) \cite{Ghosh14} \\
\hline
($\alpha_1, \alpha_2$) & Path loss exponent& ($2,3$) \cite{Ghosh14}\\
\hline
($\beta_1, \beta_2$) & Path loss at $1$ m& ($70,38$) dB \\
\hline
$\psi$ & Antenna gain& $18$ dB \\
\hline
$\tau_2$ & Time slot at $\mu$W band& $10$ ms \\
\hline
$N_0$ & Noise power& $-174$ dBm/Hz \\
\hline
\end{tabular}\label{tabsim}\vspace{-.3cm}
\end{table}}
We compare the performance with a ``CSI-based'' scheduler that relies only on the channel quality and disregards context information. The CSI-based scheduler assumes that all applications per UE can tolerate the same delay equal to the minimum of all corresponding UAs' delays. We define the performance metric $\lambda$ as the average number of QoS violations:
\begin{align}\label{lambda}
\lambda = \frac{1}{A}\sum_{t=1}^{T}\sum_{a \in \mathcal{A}_t}\mathbbm{1}_a(\pi).
\end{align}
Due to the stochastic nature of wireless channels, we are interested in the statistics of $\lambda$, i.e., to evaluate whether 
$P(\lambda \geq \lambda_{th}) \leq \epsilon,$
where $\lambda_{th}$ is the maximum tolerable $\lambda$ and $\epsilon$ is a small probability. This can be written as $F_{\lambda}(\lambda_t) \geq 1-\epsilon$, where $F_{\lambda}(.)$ is the cumulative distribution function (CDF) of $\lambda$.

Figs. \ref{sim1}(a) and \ref{sim1}(b) show the total amount of traffic transmitted at both frequency bands versus time for LoS probabilities of $\rho=0.5$ and $\rho=0.1$, respectively. Moreover, $T=5$, $M=20$, $\kappa=5$, and $B=0.5$ Mbits are considered. We observe that the traffic decreases over time, due to the finite buffer traffic model, as well as the fact that different classes of UAs are equally likely to be run at UEs. The result in Fig. \ref{sim1} shows that the proposed algorithm exploits more mmW resources as more UEs are at LoS from SBS. In addition, the dual-band scheduling significantly increases traffic offload from $\mu$W band, reaching up to $43 \%$ at the fourth time slot.     

\begin{figure}
\centering
\centerline{\includegraphics[width=8.6cm]{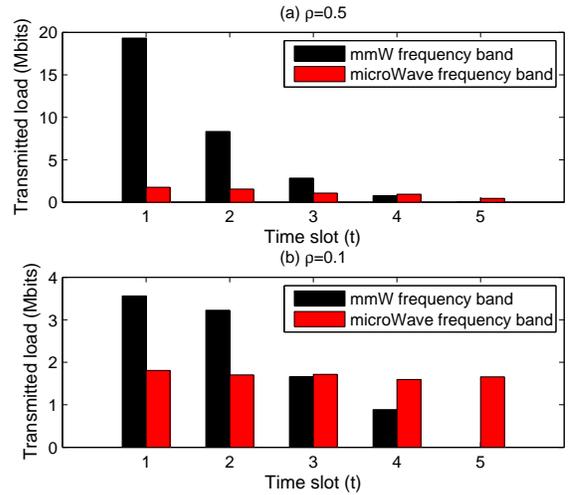}}\vspace{-0.4cm}
\caption{Comparison of the transmitted load at each frequency band vs time slot, plotted for $\rho=0.5$ and $\rho=0.1$.}\vspace{-0.5cm}
\label{sim1}
\end{figure}
\begin{figure}
\centering
\centerline{\includegraphics[width=8.6cm]{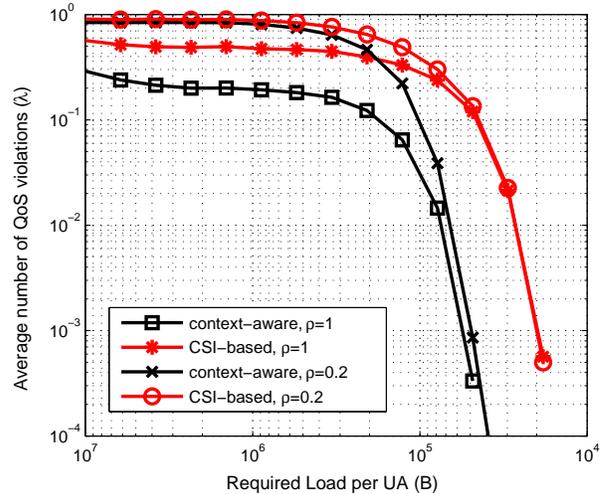}}\vspace{-0.4cm}
\caption{Performance comparison of context-aware and CSI-based scheduling.}\vspace{-0.4cm}
\label{sim3}
\end{figure}

Fig. \ref{sim3} compares the performance of the proposed context-aware approach with CSI-based scheduling for $M=10$, $\kappa=3$, $T=5$. Fig. \ref{sim3} shows the average number of QoS violations $\lambda$ as a function of the required load (number of bits $B$ that must be transmitted per UA). In Fig. \ref{sim3}, we see that $\lambda$ decreases with $B$, since more UAs meet their QoS constraint. Moreover, for $\lambda_t=0.01$, Fig. \ref{sim3} shows that the proposed approach can serve up to $60$ kbits more traffic per UA compared to the CSI-based approach. For $M=5$ and $\tau_2=10$ ms, this is equivalent to a $2.4$ Mbps improvement in the average data rate. It worth to note that for $\lambda<0.01$, the performance gain of the proposed approach against  CSI-based approach remains unchanged. In addition, for low loads, the result become independent of the probability of LoS, since UAs are scheduled in the $\mu$W band.

Fig. \ref{sim4} shows the performance comparison in terms of the CDF of $\lambda$, $F_{\lambda}(\lambda)$, for $M=20$, $\kappa=3$, $T=5$, $B=0.1$ MBits, and $\rho=0$. In Fig. \ref{sim4}(a), the $10$ farthest UEs from the SBS are chosen as cell edge UEs and the statistics of $\lambda$ for their corresponding UAs are shown. In contrast in Fig. \ref{sim4}(b), we choose the $10$ UEs that are the nearest to the SBS as cell center UEs and the statistics of $\lambda$ for their corresponding UAs are shown. With no LoS connection available for UEs, we observe that the performance of the cell edge UAs is poor for both approaches. However, Fig. \ref{sim4}(b) shows that the context-aware approach achieves more gain when cell center UAs are considered. For instance, using CSI-based approach, the probability of less than $10 \%$ of UAs be unsatisfied is $1\%$, while this probability is $37 \%$ for context-aware approach. 
\begin{figure}
\centering
\centerline{\includegraphics[width=8.6cm]{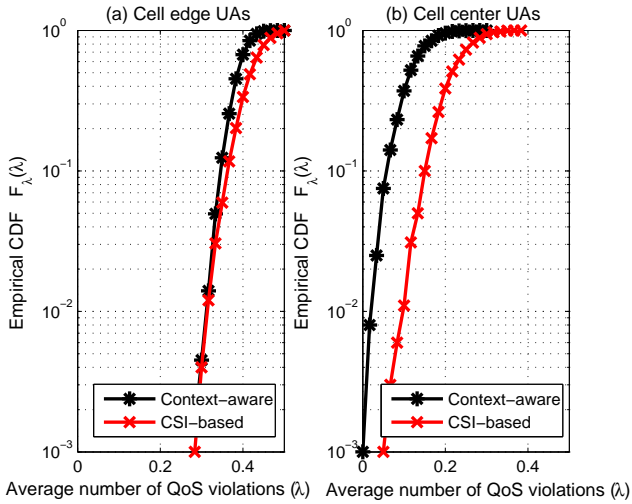}}\vspace{-0.3cm}
\caption{Performance comparison between context-aware and CSI-based scheduling for $\rho=0$, i.e., no LoS UE.}\vspace{-0.2cm}
\label{sim4}
\end{figure}

Fig. \ref{sim5} shows the same performance metric as Fig. \ref{sim4}, with the same parameters except $\rho=0.3$. Owing to available capacity at mmW band, the performance of both cell edge and cell center UAs are significantly improved compared to Fig. \ref{sim4}. Moreover, we observe that the performance gap between the proposed context-aware approach with CSI-based approach increases as $\rho$ increased. For instance, the proposed approach improves the QoS up to $90 \%$ (i.e. satisfying more UAs) for the cell edge UAs, for $F_{\lambda}(\lambda)=0.01$.
\begin{figure}
\centering
\centerline{\includegraphics[width=8.6cm]{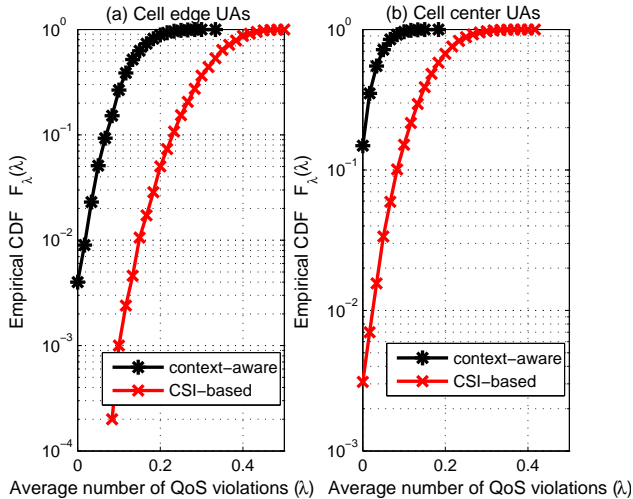}}\vspace{-0.3cm}
\caption{Performance comparison between context-aware and CSI-based scheduling for $\rho=0.3$.}\vspace{-0.4cm}
\label{sim5}
\end{figure}

The average number of iterations per time slot versus network size is shown in Fig. \ref{sim6} for $B=0.5$ Mbits, $M=5$, and $\kappa=3$. The average number of iterations increases linearly with the number of UEs. In addition, we note that the average number of iterations decreases as $\rho$ increases. This is due to the fact that more traffic can be offloaded to mmW band which enhances resource allocation at $\mu$W band. Overall, Fig. \ref{sim6} shows that the proposed algorithm converges within a reasonable number of iterations even for large number of UEs.
\begin{figure}
\centering
\centerline{\includegraphics[width=8.6cm]{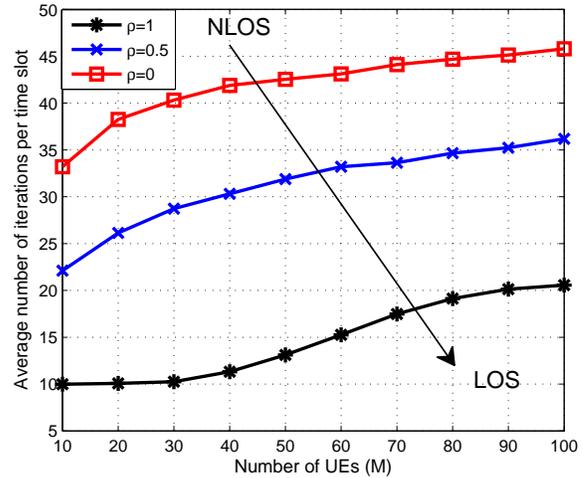}}\vspace{-0.4cm}
\caption{Number of iterations for the proposed context-aware scheduling versus the number of UEs.}\vspace{-.5cm}
\label{sim6}
\end{figure}

\section{Conclusions}
In this paper, we have proposed a novel context-aware scheduling framework for dual-mode small base stations operating at mmW and $\mu$W frequency bands. The proposed scheduler can provide delay guarantees per user application. We have formulated this context-aware scheduling problem as a one-to-many matching game which is then solved using a distributed algorithm. The proposed algorithm exploits mmW band resources for opportunistic traffic offloads, while guaranteeing the UAs' QoS. We have proved that the proposed algorithm yields a two-sided stable scheduling policy. Simulation results have shown the various merits and performance advantages of context-aware scheduling for dual-mode networks.
\bibliographystyle{ieeetr}
\bibliography{references}
\end{document}